\documentclass[letterpaper, 10 pt, conference]{ieeeconf}  

\IEEEoverridecommandlockouts                              
\usepackage[utf8]{inputenc}
\usepackage[T1]{fontenc}

\usepackage{xcolor}
\usepackage{url}
\usepackage{hyperref}
\usepackage{mathtools}

\usepackage{slashbox}
\usepackage{stfloats}
\usepackage{amsmath}
\usepackage{mathbbold}
\usepackage{amssymb}
\let\labelindent\relax
\usepackage{enumitem}
\usepackage{graphicx}
\usepackage[version=4]{mhchem}
\usepackage[ruled,vlined]{algorithm2e}

\def\dd{\text{d}}

\newtheorem{lemma}{Lemma}
\newtheorem{theorem}{Theorem}
\newtheorem{corollary}{Corollary}
\newtheorem{remark}{Remark}
\newtheorem{proposition}{Proposition}

\title{\LARGE \bf
Stochastic filters based on hybrid approximations of multiscale stochastic reaction networks
}

\author{Zhou Fang, Ankit Gupta, and Mustafa Khammash
\thanks{*This work is funded by the Swiss National Science Foundation under grant number 182653.}
\thanks{ Zhou Fang, Ankit Gupta, and Mustafa Khammash are with the Department of Biosystems Science and Engineering at ETH-Z\"urich.}%
\thanks{ \texttt{zhou.fang@bsse.ethz.ch},  \texttt{ankit.gupta@bsse.ethz.ch},  \texttt{mustafa.khammash@bsse.ethz.ch}}
}

\begin{document}

\maketitle
\thispagestyle{empty}
\pagestyle{empty}

\begin{abstract}

We consider the problem of estimating the dynamic latent states of an intracellular multiscale stochastic reaction network from time-course measurements of fluorescent reporters. We first prove that accurate solutions to the filtering problem can be constructed by solving the filtering problem for a reduced model that represents the dynamics as a hybrid process. The model reduction is based on exploiting the time-scale separations in the original network, and it can greatly reduce the computational effort required to simulate the dynamics. This enables us to develop efficient particle filters to solve the filtering problem for the original model by applying particle filters to the reduced model. We illustrate the accuracy and the computational efficiency of our approach using a numerical example.

\end{abstract}

\section{introduction}

In the past few decades, scientists' ability to look into the dynamic behaviors of a living cell has been greatly improved by the fast development of fluorescent technologies and advances in microscopic techniques.
Despite this big success, light intensity signals observed in a microscope can only report the dynamics of a small number of key components in a cell, such as fluorescent proteins and mRNAs, and, therefore, leave other dynamic states of interest, e.g., gene (on/off) state, indirectedly observable.
As a result, it is urgent to build efficient stochastic filters to estimate latent states of intracellular biochemical reaction systems from these partial observations.

Although it is possible to solve the filtering (and inference) problem accurately for intracellular systems with particle filters (see \cite{doucet2009tutorial} for particle filters), the required computational effort to simulate the system in the sampling step usually prevents it from being used on-line due to the high complexity of practical models.
Motivated by this obstacle, some methods are proposed to obtain computationally efficient filters by approximating the underlying Markov jump process by a more tractable model, such as a diffusion process \cite{chuang2013robust,golightly2006bayesian,calderazzo2019filtering} or a block model \cite{boys2008bayesian}.
All these methods are shown to be efficient for broad classes of bio-chemical reaction systems, however, their drawbacks are also clear --- the diffusion process strategy loses its validity in low copy number scenarios, and the block model methodology may not work well for highly interconnected nonlinear biological circuits.
Unfortunately, some reaction systems involve species with low copy numbers and densely interconnected biological circuits, which precludes the utilization of previous methods to these systems and require researchers to build new solutions. 

In this paper, we propose another strategy to obtain efficient particle filters based on hybrid approximations of multiscale stochastic reaction networks.
Specifically, the model reduction technique introduced in \cite{kang2013separation} is applied to approximating the underlying stochastic reaction system by a piecewise deterministic model, which can greatly reduce the computational effort required to simulate the system.
Also, we  prove that the solution to the filtering problem for the original system can be constructed by solving the filtering problem for the reduced model if some mild conditions are satisfied. 
Both of these facts enable us to develop efficient particle filters to solve the filtering problem for the original model by applying particle filters to the reduced model. 

It is worth noting that the idea of applying time-scale separation techniques to reduce underlying models to make particle filters computationally feasible has already been proposed in \cite{park2008problem,park2010dimensional,park2011particle} for diffusion type stochastic models, and its efficiency is proven in the literature \cite{imkeller2013dimensional}.
Compared with these works, our paper considers a different type of underlying model, the Markov jump process, and provides a much-simplified proof for the convergence of the approximate filters.
All these references and our paper show the efficiency of applying the time-scale separation technique to the filtering problem for multiscale systems. 

The rest of this paper is organized as follows.
In Section \ref{Sec. preliminary}, we introduce some basic notations and review the chemical reaction network theory and the filtering theory.
In Section \ref{Sec. main results}, we present our main results. 
We first prove that the accurate solution to the filtering problem can be accurately approximated by the solution to the filtering problem for the reduced model under some mild condition.
Then, based on this result, we establish an efficient particle filter to solve the filtering problem for the original model by applying the particle filter to the reduced model.
For the sake of readability, we put the proof of our main result in the appendix.
In Section \ref{Sec. numerical simulation}, a numerical example is presented to illustrate our approach.
Finally, Section \ref{Sec. conclusion} concludes this paper.

\section{Preliminary} {\label{Sec. preliminary}}

\subsection{Notations}
We first denote a natural filtered probability space by $( \Omega,  \mathcal F, \{\mathcal F_{t}\}_{t\geq 0}, \mathbb P)$, where $\Omega$ is the sample space, $\mathcal F$ is the $\sigma$-algebra, $\{\mathcal F_{t}\}_{t\geq 0}$ is the filtration, and $\mathbb P$ is the natural probability.
Also, we term $\mathbb N_{>0}$ as the set of positive integers, $\mathbb R^{n}$ with $n$ being a positive integer as the space of $n$-dimensional real vectors, $\|\cdot\|$ as the Euclidean norm, $|\cdot|$ as the absolute value notation, $s\wedge t$ (for any $s,t \in \mathbb{R}$ ) as $\min(s,t)$, and $s\vee t$ as $\max(s,t)$.
For any positive integer $n$ and any $t>0$, we term $\mathbb D_{ \mathbb R^{ n}}[0,t]$ as the Skorokhod space that consists of all $\mathbb R^n$ valued cadlag functions on $[0,t]$ and $\mathbb D_{ \mathbb R^{ n}}[0,\infty)$ as the Skorokhod space that consists of all $\mathbb R^n$ valued cadlag functions on $[0,\infty)$.

\subsection{Stochastic chemical reaction networks} 

We consider an intracellular system undergoing reactions
\begin{equation*}
v_{1,j} S_1 + \dots + v_{n,j} S_n \ce{ ->[$k_{j}$]} v'_{1,j} S_1 + \dots + v'_{n,j} S_n, ~~ j=1,\dots, r,
\end{equation*}
where $S_i$ ($i=1,\dots,n$) are distinguished species in the systems, $v_{i,j}$ and $v'_{i,j}$ are non-negative integers, called stoichiometric coefficients, $k_{j}$ are the reaction constants, and $r$ is the number of reactions.
Also, we name the linear combination of species (e.g, $v_{1,j} S_1 + \dots + v_{n,j} S_n$).
The topology of the above reaction equations can be fully represented by a triplet $\{\mathcal S, \mathcal C, \mathcal R\}$, called a chemical reaction network \cite{feinberg2019foundations}, where
\begin{itemize}
	\item $\mathcal{S}$ is the species set $\{S_1, S_2, \dots, S_n\}$ whose elements represent distinct substances,
	\item $\mathcal C$ is the complex set $\bigcup_{j=1}^{r} \{v_{\cdot j} , v'_{\cdot j} \}$, where $v_{\cdot j}\triangleq \left(v_{1,j}, \dots, v_{n,j} \right)^{\top}$ and $v'_{\cdot j}\triangleq \left(v'_{1,j}, \dots, v'_{n,j} \right)^{\top}$, indicating substrate complexes and product complexes,
	\item $\mathcal R$ is the reaction set $\{v_{\cdot,1} \to v'_{\cdot1}, \dots, v_{\cdot r} \to v'_{\cdot r}\}$ showing the connections between the complexes.
\end{itemize}
Let $X(t)=\left(X_1(t),X_2(t), \dots, X_n(t)\right)^{\top}$ be the numbers of molecules of these species at time $t$, then the system's dynamics following mass-action kinetics can be expressed as \cite{anderson2011continuous}
\begin{equation*}
X(t)= X(0)+\sum_{j=1}^{r} (v'_{\cdot j}-v_{\cdot j}) R_j\left( \int_0^t \lambda_j (X(s))\dd s \right)
\end{equation*}
where $R_{j}(t)$ ($j=1,\dots,r$) are mutually independent unit rate Poisson random processes, $\lambda_j (x)\triangleq k_{j} \prod_{i=1}^n \frac{x_i!}{\left(x_i-v_{i,j}\right)!}\mathbbold 1_{\{x_i\geq v_{i,j}\}}$ with $\mathbbold{1}_{\{\cdot\}}$ the indicator function, and the initial condition $X(0)$ is subject to a particular known distribution.

\subsection{Model reduction via the time scale separation}
In a practical bio-chemical reaction system, different species can vary a lot in abundance, and rate constants can also vary over several orders of magnitude. 
To normalize these quantities, we term $N$ as a scaling factor, $\alpha_i$ (for $i=1,\dots, n$) as the magnitude of the $i$-th species such that $X^{N}_i(t)\triangleq N^{-\alpha_i}X_{i}(t) = O(1)$, $\beta_j$ (for $j=1,\dots,r$) as the magnitude of the reaction constant $k_{j}$ such that $k'_j\triangleq k_j N^{-\beta_j}=O(1)$,
$\gamma$ as the time scale parameter, and $X^{N,\gamma}(t) \triangleq X^{N}(tN^{\gamma})$.
Then, in the new coordinate, the dynamic equation can be expressed as \cite{kang2013separation}
\begin{align}{\label{eq. scaling stochastic dynamics}}
X^{N,\gamma}(t) =& X^{N,\gamma}(0)   +\sum_{j=1}^{r} \Lambda^{N} (v'_{\cdot j}-v_{\cdot j}) \\
& \qquad  \qquad \quad \times R_j\left( \int_0^{t} N^{\gamma+\tilde \rho_j} \lambda^N_{j}(X^{N,\gamma}(s)) \dd s \right) \notag
\end{align}
where $\Lambda^{N}= \text{diag}(N^{-\alpha_1},\dots,N^{-\alpha_s})$, $\alpha=(\alpha_1,\dots,\alpha_n)$, $\lambda^{N}_j (x)\triangleq k'_j \prod_{i=1}^{n} \prod_{\ell=0}^{v_{i,j}-1}  (x_{i}-\ell N^{-\alpha_i})\mathbbold 1_{\left\{x_i\geq N^{-\alpha_i}v_{i,j}\right\}} $ are of order $O(1)$, and $\tilde \rho_j=\beta_j + v^{\top}_{\cdot j} \alpha$.
For this system, we denote its transition kernel by $K^{N,\gamma_1}_{t}(x,A)\triangleq \mathbb P ( X^{N,\gamma_1}(t)\in A | X^{N,\gamma_1}(0)=x )$.

Following the notations in \cite{kang2013separation}, we further term $\Gamma_i^{+}\triangleq \{ j | v'_{i,j}>v_{i,j} \}$ as the set of reactions that result in an increase of the $i$-th species, $\Gamma_i^{-}\triangleq \{ j | v'_{i,j}<v_{i,j} \}$ as the set of reactions that result in a decrease of the $i$-th species.
Then, the constant $\alpha_i-\max_{j\in \Gamma_i^{+}\cup \Gamma_i^{-}} \left(\beta_j+v_{\cdot j}^{\top} \alpha \right)$ is the time scale of the $i$-th substance, i.e., the time scale where $X^{N,\gamma}_i (t)$ evolves at the rate of $O(1)$.
Moreover, we term $\gamma_1\triangleq \min_i
\left( \alpha_i-\max_{j\in \Gamma_i^{+}\cup \Gamma_i^{-}} \left(\beta_j+v_{\cdot j}^{\top} \alpha \right)\right)$ as the parameter of the fastest time scale of the system and
$D^{\tilde \alpha} \triangleq \text{diag}(\dots \mathbbold 1_{\{\alpha_i=\tilde \alpha\}}\dots)$ as a diagonal matrix indicating whether a species is at the scale of $N^{\tilde \alpha}$.
In the rest of this paper, we focus on the dynamic behaviors of a reaction system at the fastest time scale $\gamma_1$.

For a large scaling factor $N$, by neglecting all slow reactions ($\gamma_1+\tilde \rho_j<0$) and approximating fast reactions ($\gamma_1+\tilde \rho_j>0$) by a continuous process, we can arrive at a simplified piecewise deterministic process as follows.
\begin{small}
\begin{align}
X^{\gamma_{1}}(t)=& \lim_{N\to \infty} X^{N}(0) \\
&
+ \sum_{j : \gamma_1+\tilde \rho_j>0} \int_0^t \lambda'_j( X^{\gamma_1}(s))  D^{\gamma_1+\tilde \rho_j}\left(v'_{\cdot j}-v_{\cdot j}\right)\dd s  \notag\\
& +  \sum_{j : \gamma_1+\tilde \rho_j=0 }
R_{j}\left(\int_0^t \lambda'_j( X^{\gamma_1}(s)) \dd s \right) D^0\left(v'_{\cdot j}-v_{\cdot j}\right). \notag
\end{align}
\end{small}We denote the transition kernel of the above process by $K^{\gamma_1}_{t}(x,A)\triangleq \mathbb P ( X^{\gamma_1}(t)\in A | X^{\gamma_1}(0)=x )$.
If we further assume the initial conditions to satisfy that 
\begin{equation}{\label{eq. assumption initial conditons}}
\lim_{N\to \infty} X^{N}(0) \text{ exists,}
\end{equation}
and that the above systems are non-explosive, i.e.,
\begin{equation}{\label{eq. assumption infinite explosion time gamma1}}
\lim_{c\to\infty}\tau^{N,\gamma_1}_{c}=\infty,  \quad \text{ and } \quad
\lim_{c\to\infty} \tau^{\gamma_1}_{c}=\infty,   \quad  \text{a.s.},
\end{equation}
where $\tau^{N,\gamma_1}_{c}\triangleq \inf \left\{t~|~\|X^{N,\gamma_1}(t)\|\geq c \right\}$ and $ \tau^{\gamma_1}_{c}\triangleq \inf \left\{t~|~ \|X^{\gamma_1}(t)\|\geq c\right\}$,
then the original model $X^{N,\gamma_1}(\cdot)$ converges in distribution to the reduced model $X^{\gamma}(\cdot)$ on any finite time interval. 

\begin{proposition}[Adapted from \cite{kang2013separation}]{\label{cor convergence on any finite interval}}
	On any finite time interval $[0,t]$, if \eqref{eq. assumption initial conditons} and \eqref{eq. assumption infinite explosion time gamma1} hold, then $X^{N,\gamma_1}(\cdot )\Rightarrow X^{\gamma_{1}}(\cdot)$ in the sense of the Skorokhod topology.
\end{proposition}
\begin{proof}
	It follows from \cite[Theorem 4.1]{kang2013separation} and the portmanteau lemma.
\end{proof}

Also, the finite-dimensional distribution of $X^{N,\gamma_1}(\cdot)$ converges to the finite-dimensional distribution of $X^{\gamma_1}(\cdot)$.

\begin{corollary}\label{cor convergence on discrete time point}
	Let $i$ be a positive integer and $\{t_j\}_{0<j\leq i}$ be a sequence of scalars. If \eqref{eq. assumption initial conditons} and \eqref{eq. assumption infinite explosion time gamma1} hold, then $X^{N,\gamma_1}_{t_1:t_i} \Rightarrow X^{\gamma_1}_{t_1:t_i}$ where $  X^{N,\gamma_1}_{t_1:t_i} \triangleq \left(X^{N,\gamma_1}(t_1), \dots, X^{N,\gamma_1}(t_i)\right)$ and $ X^{\gamma_1}_{t_1:t_i} \triangleq  \left(X^{\gamma_1}(t_1), \dots, X^{\gamma_1}(t_i)\right)$.
\end{corollary}
\begin{proof}
	It follows immediately from Proposition \ref{cor convergence on any finite interval} and \cite[Theorem 7.8 in Chapter 3]{ethier1986markov}.
\end{proof}

Note that the convergence result does not always hold on the infinite time horizon, because the limit model may miss bi-stability or some other phenomenon that the full model has (see \cite[Section VI]{gillespie2000chemical}).

In the above results, the non-explosivity condition \eqref{eq. assumption infinite explosion time gamma1} is not a very strong condition; there exist checkable sufficient conditions to ensure non-explosivity (see \cite{gupta2014scalable,anderson2018non}) that can deal with it efficiently.
In this paper, we do not provide specific conditions for non-explosivity in order to make the result general and avoid the deviation from this paper's main concern.

\begin{remark}{\label{remark compare computation load of the original model and the reduced one}}
	The computational complexity to simulate $X^{\gamma_1}(\cdot)$ can be greatly lower than the complexity to simulate $X^{N, \gamma_1}(\cdot)$, because the former avoids the exact simulation of fast reactions ($\gamma+\tilde \rho_j >0$), which consume a lot of computational resources to update the system at a rate proportional to $N^{\gamma_1+\tilde \rho_j}$.
\end{remark}

\subsection{Filtering problems, the change of measure method, and the particle filter.} 

In this paper, we assume that $m$ channels of light intensity signals of fluorescent reporters in a single cell can be observed discretely in time from a microscopy platform (e.g., the one in \cite{rullan2018optogenetic}).
The discrete-time property is mainly caused by the cell segmentation and tracking, which is required prior to the measuring step and can be time-consuming. 
Suppose that $\{t_{i}\}_{i\in\mathbb N_{>0}}$ are time points when the observation comes, and the value of each observation, $Y^{N, \gamma_1}(t_i)$, satisfies 
\begin{equation*}
	Y^{N,\gamma_1}_{\ell}(t_i)= h_{\ell}\left(X^{N,\gamma_1}(t_i)\right)  + W_{\ell}(t_i) \qquad \ell=1,\dots,m,
\end{equation*}
where $h_{\ell}$ are bounded Lipschitz continuous functions indicating the relation between the observation and the reaction process, and $W_{\ell}(t_i)$ ($\ell =1,\dots,m$ and $i\in \mathbb N_{>0}$) are mutually independent standard Gaussian random variables which are also independent of $R_{j}(\cdot)$ ($j=1,\dots, r$).

Our task is to provide optimal estimates to latent states of the system, e.g., $\phi(X^{N,\gamma_1}(t_i))$ where $\phi$ is a known measurable function, in the sense of mean square error based on observations up to the time $t_i$; in other words, we are going to calculate the conditional expectation $\pi^{N,\gamma_1}_{t_i}(\phi) \triangleq \mathbb E _{\mathbb P}\left[ \phi(X^{N,\gamma_1}(t_i)) ~|~ \mathcal Y^{N,\gamma_1}_{t_i} \right]$, where $\mathcal Y^{N,\gamma_1}_{t_i}$ is the filtration generated by the process $\{Y^{N,\gamma_1}(t_j)\}_{0 < j\leq t_i}$.

Similarly, we can define artificial readouts for the reduced model by 
\begin{equation*}
		Y^{\gamma_1}_{\ell}(t_i)= h_{\ell}\left(X^{\gamma_1}(t_i)\right)  + W_{\ell}(t_i) \qquad \ell=1,\dots,m,
\end{equation*}
and attempt to solve the corresponding filtering problem, i.e., calculating the conditional expectation $\pi^{\gamma_1}_{t_i}(\phi) \triangleq \mathbb E _{\mathbb P}\left[ \phi(X^{\gamma_1}(t_i)) ~|~ \mathcal Y^{\gamma_1}_{t_i} \right]$, where $\mathcal Y^{\gamma_1}_{t_i}$ is the filtration generated by the process $\{Y^{\gamma_1}(t_j)\}_{0 < j\leq t_i}$.
Note that in practical experiments, one can only observe $\{Y^{N,\gamma_1}_{\ell}(t_i)\}_{i\in\mathbb{ N}_{>0}}$ but never get $\{Y^{\gamma_1}_{\ell}(t_i)\}_{i\in\mathbb{ N}_{>0}}$.

The change of measure method is a powerful tool to deal with filtering problems, whose core idea lies in constructing a reference probability measure that orthogonalizes the underlying process and the observation.
For our problem, we construct an auxiliary random process $Z^{N,\gamma_1}(t_i)= \prod_{j=1}^{i} g\left(X^{N,\gamma_1}(t_j),Y^{N,\gamma_1}(t_j)\right)$, where $g(x,y)=\exp\left(\sum_{\ell=1}^{m} h_{\ell}(x) y_{\ell}-{h^2_{\ell}(x)}/{2}\right)$, and whose reciprocal is a martingale with respect to the filtration $\mathcal F_{t_i}$ under $\mathbb P$.
Based on this martingale, we define a reference probability measure, $\mathbb{ P}^{N,\gamma_1}$, by the Radon-Nikodym derivation $\left.\frac{\dd \mathbb{ P}^{N,\gamma_1}}{\dd \mathbb P}\right|_{\mathcal F_{t_i}}= \left(Z^{N,\gamma_1}(t_i)\right)^{-1}$, under which $\{Y^{N,\gamma_1}(t_i)\}_{i\in \mathbb N_{>0}}$ are mutually independent standard m-dimensional Gaussian random variables, the law of the process $X^{N,\gamma_1}(\cdot)$ is the same as its law under $\mathbb{ P}$, and the process $X^{N,\gamma_1}(\cdot)$ and observation $Y^{N,\gamma_1}(\cdot)$ are independent of each other. (These results can be easily shown by checking the joint characteristic function of these processes.)
Furthermore, by the Kallianpur-Striebel formula \cite[Theorem 3]{kallianpur1968estimation}, for any measurable function $\phi$ such that $\phi(X^{N,\gamma_1}(t_i))$ is integrable under $\mathbb{ P}$, there holds
\begin{align}{\label{eq. unnormalized filter to normalized filter}}
\pi^{N,\gamma_1}_{t_i}(\phi) = \frac{{\mathbb E}_{\mathbb P^{N,\gamma_1}} \left[\left. Z^{N,\gamma_1}(t_i) \phi\left(X^{N,\gamma_1}(t_i)\right) \right|  \mathcal Y^{N,\gamma_1}_{t_i}  \right]}{ {\mathbb E} _{\mathbb P^{N,\gamma_1}} \left[ \left. Z^{N,\gamma_1}(t_i) \right|  \mathcal Y^{N,\gamma_1}_{t_i}  \right]}
\end{align}
$\mathbb P$\text{-a.s.} and $\mathbb P^{N,\gamma_1}$\text{-a.s.}.
Recall that $\pi^{N,\gamma_1}_{t_i}(\phi)$ is the notation for $\mathbb E _{\mathbb P}\left[ \phi(X^{N,\gamma_1}(t_i)) ~|~ \mathcal Y^{N,\gamma_1}_{t_i} \right]$.
Here, we name $\rho^{N,\gamma_1}_{t_i} (\phi)\triangleq{\mathbb E}_{\mathbb P^{N,\gamma_1}} \left[\left. Z^{N,\gamma_1}(t_i) \phi\left(X^{N,\gamma_1}(t_i)\right) \right|  \mathcal Y^{N,\gamma_1}_{t_i}  \right]$ as the unnormalized conditional expectation and $ \rho^{N,\gamma_1}_{t_i} (1)$ as the normalization factor.
Note that both $\pi^{N,\gamma_1}_{t_i}(\phi)$ and $\rho^{N,\gamma_1}_{t_i}(\phi)$ are $\mathcal Y^{N,\gamma_1}_{t_i}$-measurable.
Therefore, for any bounded $\phi(\cdot)$, there exist measurable functions, $\hat f^{N,\gamma_1}_{\phi,t_i}(\cdot)$ and $\hat g^{N,\gamma_1}_{\phi,t_i}(\cdot)$ from the domain of the observations to $\mathbb R$ such that $\hat f^{N,\gamma_1}_{\phi,t_i}\left(Y^{N,\gamma_1}_{t_1:t_i}\right)=\pi^{N,\gamma_1}_{t_i}(\phi)$ and $\hat g^{N,\gamma_1}_{\phi,t_i}\left(Y^{N,\gamma_1}_{t_1:t_i}\right)=\rho^{N,\gamma_1}_{t_i}(\phi)$, where $Y^{N,\gamma_1}_{t_1:t_i}\triangleq \left(Y^{N,\gamma_1}(t_1),\dots,Y^{N,\gamma_1}(t_i)\right)$ is the observation process up to the time $t_i$.

Based on \eqref{eq. unnormalized filter to normalized filter}, an algorithm called the particle filter (also known as sequential Monte Carlo method) can be constructed as Algorithm \ref{alg particle filter} to numerically solve the filtering problem. 
In Algorithm \ref{alg particle filter}, the particles $x_j(\cdot)$ mimic the dynamic state $X^{N,\gamma_1}(\cdot)$, and the weights $w_j(\cdot)$ mimic the normalized $Z^{N,\gamma_1}(\cdot)$; therefore, the computed filter can approximate the true filter by \eqref{eq. unnormalized filter to normalized filter}.
A resampling step is executed in each iteration to remove non-significant particles so that sample impoverishment can be avoided \cite{doucet2009tutorial}.

\begin{algorithm}
	\SetAlgoLined
	Input a kernel function $K_t(\cdot,\cdot)$ and observations $Y(\cdot)$\;
	Initialize $M$ particles $ x_1(0),\dots, x_M(0)$ to $X^{N}(0)$ and weights $ w_{1}(0), \dots, w_M(0)$ to $1/M$. Set $t_0=0$\;
	\For{each time point $t_i$ ($i\in\mathbb N_{>0}$)}{
		For $j=1,\dots,M$, sample $x_j(t_i)$ from $K_{t_{i+1}-t_{i}}( x_{j}(t_{i-1}), \cdot)$ and compute weights $w_{j}(t_i)\propto w_{j}(t_{i-1}) g\left(x_{j}(t_i), Y(t_i)\right)$\;
		Compute filter $\bar \pi_{M,t_{i}}(\phi)=  \sum_{j=1}^{M} w_j(t_i)\phi \left(x_{j}(t_i)\right)$\;
		Resample $\{w_{j}(t_i), x_j(t_{i})\}$ to obtain $M$ equally weighted particles 
		$\{{1}/{M}, \bar x_j(t_{i})\}$ and set $\{w_{j}(t_i), x_j(t_{i})\}\leftarrow \{{1}/{M}, \bar x_j(t_{i})\}$\;
	}
	\caption{The particle filter \cite{doucet2009tutorial}}
	\label{alg particle filter}
\end{algorithm}


Similarly, for the filtering problem of the reduced model, we can define another auxiliary function $Z^{\gamma_1}(t_i)= \prod_{j=1}^{i} g\left(X^{\gamma_1}(t_j),Y^{\gamma_1}(t_j)\right)$ and a reference probability $\mathbb{ P}^{\gamma_1}$ by $\left.\frac{\dd \mathbb{ P}^{\gamma_1}}{\dd \mathbb P}\right|_{\mathcal F_{t_i}}= \left(Z^{\gamma_1}(t_i)\right)^{-1}$, under which $\{Y^{\gamma_1}(t_i)\}_{i\in \mathbb N_{>0}}$ are mutually independent standard m-dimensional Gaussian random variables, the law of the process $X^{\gamma_1}(\cdot)$ is the same as its law under $\mathbb{ P}$, and the process $X^{\gamma_1}(\cdot)$ and observation $Y^{\gamma_1}(\cdot)$ are independent of each other.
By the Kallianpur-Striebel formula, for any measurable function $\phi$ such that $\phi(X^{\gamma_1}(t_i))$ is integrable under $\mathbb{ P}$, there holds
\begin{align}{\label{eq. unnormalized filters reduced model}}
\pi^{\gamma_1}_{t_i}(\phi) = \frac{{\mathbb E}_{\mathbb P^{\gamma_1}} \left[\left. Z^{\gamma_1}(t_i) \phi\left(X^{\gamma_1}(t_i)\right) \right|  \mathcal Y^{\gamma_1}_{t_i}  \right]}{ {\mathbb E} _{\mathbb P^{\gamma_1}} \left[ \left. Z^{\gamma_1}(t_i) \right|  \mathcal Y^{\gamma_1}_{t_i}  \right]}
\end{align}
$\mathbb P\text{-a.s.}$ and $\mathbb P^{\gamma_1}\text{-a.s.}$, in which we further denote the numerator by $\rho^{\gamma_1}_{t_i} (\phi)$ and the denominator by $\rho^{\gamma_1}_{t_i} (1)$.
Note that both $\pi^{\gamma_1}_{t_i}(\phi)$ and $\rho^{\gamma_1}_{t_i}(\phi)$ are $\mathcal Y^{\gamma_1}_{t_i}$-measurable.
Therefore, for any bounded $\phi(\cdot)$, there exist measurable functions, $\hat f^{\gamma_1}_{\phi,t_i}(\cdot)$ and $\hat g^{\gamma_1}_{\phi,t_i}(\cdot)$ from the domain of the observations to $\mathbb R$, such that $\hat f^{\gamma_1}_{\phi,t_i}\left(Y^{\gamma_1}_{t_1:t_i}\right)=\pi^{\gamma_1}_{t_i}(\phi)$ and $\hat g^{\gamma_1}_{\phi,t_i}\left(Y^{\gamma_1}_{t_1:t_i}\right)=\rho^{\gamma_1}_{t_i}(\phi)$, where $Y^{\gamma_1}_{t_1:t_i}\triangleq \left(Y^{\gamma_1}(t_1),\dots,Y^{\gamma_1}(t_i)\right)$ is the observation process up to the time $t_i$.



\section{Main results} {\label{Sec. main results}}

Recall that our task is to accurately and computationally efficiently solve the filtering problem for the multi-scale system \eqref{eq. scaling stochastic dynamics}, i.e., to calculate $\{\pi^{N,\gamma_1}_{t_{i}}(\phi)\}_{i\in\mathbb N_{>0}}$.
A straightforward idea to solve this problem is to use the particle filter where the transition kernel of the full model (i.e., $K^{N,\gamma_1}(\cdot,\cdot)$) and observations $Y^{N,\gamma_1}(\cdot)$ are inserted. 
We denote this particle filter by $\bar \pi^{N,\gamma_1}_{M,t_{i}}(\phi)$.
Although it can accurately solve the filtering problem if the resampling method is properly chosen (see \cite[Corollary 2.4.4]{crisan2001particle}), this algorithm involves the simulation of the full model \eqref{eq. scaling stochastic dynamics} and, therefore, can be computationally expensive (see Remark \ref{remark compare computation load of the original model and the reduced one}).
Consequently, we need to figure out a smarter way to approach this problem.

We first show that the exact solution to the filtering problem of the original model can be constructed by solving the filtering problem for the reduced model. 

\begin{theorem}{\label{thm convergence at the first time scale}}
	Assume that \eqref{eq. assumption initial conditons} and \eqref{eq. assumption infinite explosion time gamma1} are satisfied.
	Then, for any $i\in\mathbb N_{>0}$ and any bounded continuous function $\phi(\cdot)$, 
	there holds  $\pi^{N,\gamma_1}_{t_i}(\phi)-\tilde \pi^{N,\gamma_1}_{t_i}(\phi) \triangleq
	\hat{f}_{\phi,t_i}^{N,\gamma_1}(Y^{N,\gamma_1}_{t_1:t_i}) - \hat{f}_{\phi,t_i}^{\gamma_1}(Y^{N,\gamma_1}_{t_1:t_i}) \to 0$, as $N$ goes to infinity, in $\mathbb{ P}$-probability.
\end{theorem}
\begin{proof}
	The proof is in the appendix.
\end{proof}

Recall that $\hat{f}_{\phi,t_i}^{\gamma_1}(\cdot)$ is a mapping that maps the observation to the solution of the filtering problem for the reduced model.  
The above theorem suggests that the probability of the error of these filters below a given threshold is very close to 1 if the scaling factor, $N$, is large, and, therefore, that $\tilde \pi^{N,\gamma_1}_{t_i} (\phi)\triangleq\hat{f}_{\phi,t_i}^{\gamma_1}(Y^{N,\gamma_1}_{t_1:t_i})$ is a good approximation to the exact filter $ \pi^{N,\gamma_1}_{t_i} (\phi)$.


Let us constructed a particle filter $\bar \pi^{\gamma_1}_{M,t_{i}}(\phi)$ where the transition kernel of the reduced model (i.e., $K^{\gamma_1}(\cdot,\cdot)$) and observations $Y^{N,\gamma_1}(\cdot)$ are inserted.
Clearly, this particle filter is an approximation of the filter $\hat{f}_{\phi,t_i}^{\gamma_1}(Y^{N,\gamma_1}_{t_1:t_i})$.
Based on Theorem \ref{thm convergence at the first time scale}, we can further show that this particle filter $\bar \pi^{\gamma_1}_{M,t_{i}}(\phi)$ can also accurately approximate the true filter $\pi^{N,\gamma_1}_{t_i}(\phi)$ with high probability.


\begin{corollary}{\label{cor particle filter}}
	Assume that \eqref{eq. assumption initial conditons} and \eqref{eq. assumption infinite explosion time gamma1} are satisfied, and the resampling step in the particle filter (Algorithm \ref{alg particle filter}) adopts the multinomial branching method.
	Then for any $i\in\mathbb N_{>0}$, any positive $\delta$, and any bounded continuous function $\phi(\cdot)$, there holds
	\begin{equation*}
	 \lim_{N\to\infty} \lim_{M\to\infty} \mathbb P \left(\left|\bar \pi^{\gamma_1}_{M,t_{i}}(\phi)-\pi^{N,\gamma_1}_{t_{i}}(\phi)\right| >\delta \right) =0.
	\end{equation*}
\end{corollary}
\begin{proof}
	It is a direct consequence of Corollary 2.4.4 in the literature \cite{crisan2001particle} and Theorem \ref{thm convergence at the first time scale}.
\end{proof}

Recall that the required computational effort to simulate the reduced model $X^{\gamma}(\cdot)$ is much lower than the effort to simulate $X^{N,\gamma}(\cdot)$ (see Remark \ref{remark compare computation load of the original model and the reduced one}).
As a consequence, particle filter $\bar \pi^{\gamma_1}_{M,t_{i}}(\phi)$, which applies the kernel function $K^{\gamma_1}(\cdot,\cdot)$ and, therefore, is required to simulate the reduced model $X^{\gamma_1}(\cdot)$ in the sampling step, has a much lower computational cost compared with the particle filter $\bar \pi^{N,\gamma_1}_{M,t_{i}}(\phi)$, which applies kernel function $K^{N, \gamma_1}(\cdot,\cdot)$ and is required to simulated the full model.
This fact, together with the accuracy result Corollary \ref{cor particle filter}, indicates that the filter $\bar \pi^{\gamma_1}_{M,t_{i}}(\phi)$ based on the reduced model is a better candidate than $\bar \pi^{N,\gamma_1}_{M,t_{i}}(\phi)$ for solving a practical filtering problem for multiscale intracellular reaction systems.


We need to emphasize that the boundedness condition of $\phi$ is not generally satisfied in biological applications, where concentrations of most target species have no theoretical upper bounds apart from some exceptions, e.g., gene copies. To tackle this problem, we can truncate a target quantity by a large number beyond which the conditional probability is comparatively low, and the tail event contributes little to the conditional expectation. Then, an estimation of the truncated quantity generated by the particle filtering applying the reduced model can provide an accurate approximation to the conditional expectation of the target quantity.

\section{Numerical example}{\label{Sec. numerical simulation}}

In this section, we illustrate our approach using a numerical example.
We consider a gene expression model involving four species: a gene having on and off states (denoted respectively as $S_2$ and $S_1$), the mRNA it transcribes (denoted as $S_3$), and a fluorescent protein it expresses (denoted as $S_4$), and six reactions: $S_1 \ce{->[$k_1$]} S_2 $ (gene activation), $S_2 \ce{->[$k_2$]} S_1$  (gene deactivation), $S_2 \ce{->[$k_3$]} S_3 + S_2$ (gene transcription), $S_3	\ce{ ->[$k_{4}$]}  S_3+S_4$ (translation), $S_3 \ce{->[{$k_5$}]} \emptyset$ (mRNA degradation), and $S_4 \ce{ ->[$k_{6}$]} \emptyset $ (protein degradation).

For this system, we take $N=100$, $\alpha_4=1$, and $\alpha_i=0$ for $1\leq i \leq 3$, i.e., the cellular system consists of hundreds of fluorescent protein molecules but very few copies of other molecules. 
The values of reaction constants and their scaling exponents are shown in Table \ref{table reaction constants}, where the unit of both reaction constants and scale rates is ``minute$^{-1}$".
Also, we assume $X^N_{1}(0)$ to have a binary distribution with mean $1/3$, $X^{N}_2(0)$ to satisfy $X^{N}_2(0)=1- X^{N}_1(0)$, $X^{N}_3(0)$ to have a Poisson distribution with mean 2, and $X^{N}_4(0)$ to also have a Poisson distribution with mean 2.
Moreover, we assume that all initial conditions except $X_2(0)$ are independent of each other.
\begin{table}[h]
	\begin{tabular}{|cr|cr|cr|cr|}
		\hline
		\multicolumn{2}{c}{Reaction constants} & \multicolumn{2}{c}{Exponents} &  \multicolumn{2}{c}{Scaled rates} &  \multicolumn{2}{c}{Reaction scale}  \\
		\hline
		$k_1$ & $1.40 \times 10^{-2}$  & $\beta_1$  &  0&   $k'_1$ & 0.0140 &  $\beta_{1}+ \alpha_1$ & 0 \\	$k_2$ &  $8.40\times 10^{-3}$ & $\beta_2$ &  $0$&  $k'_2$& 0.00840 & $\beta_{2} + \alpha_2$ & 0 \\
		$k_3$ &  $7.15\times 10^{-1}$&  $\beta_3$&  0&  $k'_3$& 0.715& $\beta_{3} + \alpha_2$ & 0 \\
		$k_4$ & $3.90\times 10^{+1}$ & $\beta_4$ &  1 &  $k'_4$& 0.390& $\beta_{4} + \alpha_3$ & 1 \\
		$k_5$ &  $1.99\times 10^{-1}$&  $\beta_5$&  $0$&  $k'_5$& 0.199 & $\beta_{5} + \alpha_3$& 0\\
		$k_6$ &  $3.79\times 10^{-1}$&  $\beta_6$&  0&  $k'_6$& 0.379& $\beta_{6} + \alpha_4$ & 1 \\
		\hline
	\end{tabular}
	\caption{Scaling exponents for reaction rates}
	\label{table reaction constants}
\end{table}
With this setting, we can easily calculate that $\gamma_1=0$, i.e., the fastest time scale is 0, and the full dynamic model satisfies
\begin{small}
\begin{align*}
	&X^{N,\gamma_1}(t) \\
	&= X^{N,\gamma_1}(0) + 
	\zeta_1
	R_1 \left( k'_1  \int_0^t X^{N,\gamma_1}_{1} (s) \dd s \right) \\
	&\quad + 
    \zeta_2
	R_{2}\left(k'_{2} \int_0^t X^{N,\gamma_1}_2(s) \dd s\right)
	+		\zeta_3
	R_{3} \left(k'_3\int_0^t X^{N,\gamma_1}_2(s) \dd s\right)\\
	&\quad+
	\zeta_4N^{-1}R_4\left(k'_4  N \int_0^t X^{N,\gamma_1}_3(s) \dd s\right)
    +
	\zeta_5 R_5\left(k'_5\int_0^t X^{N,\gamma_1}_{3}(s) \dd s\right) \\
	&\quad +
	\zeta_6 N^{-1}R_6\left(k'_6 N  \int_0^t X^{N,\gamma_1}_4(s) \dd s\right) 
\end{align*}
\end{small}where $\zeta_1=(-1,1,0,0)^{\top}$, $\zeta_2=(1,-1,0,0)^{\top}$, $\zeta_3=(0,0,1,0)^{\top}$, $\zeta_4=(0,0,0,1)^{\top}$, $\zeta_5=(0,0,-1,0)^{\top}$, and $\zeta_6=(0,0,0,-1)^{\top}$;
the reduced dynamics satisfies
\begin{small}
\begin{align*}
	&X^{\gamma_1}(t) \\
	&= \lim_{N\to \infty}X^{N,\gamma_1}(0)+ 
		R_1 \left( k'_1  \int_0^t X^{N,\gamma_1}_{1} (s) \dd s \right) \\
	&\quad + 
	\zeta_2
	R_{2}\left(k'_{2} \int_0^t X^{N,\gamma_1}_2(s) \dd s\right)
	+		\zeta_3
	R_{3} \left(k'_3\int_0^t X^{N,\gamma_1}_2(s) \dd s\right)\\
	&\quad + \zeta_4\int_0^t k'_4 X^{N,\gamma_1}_3(s) \dd s  +
	\zeta_5 R_5\left(k'_5\int_0^t X^{N,\gamma_1}_{3}(s) \dd s\right) \\
	&\quad 
	 +
	\zeta_6 \int_0^t k'_6  X^{N,\gamma_1}_4(s) \dd s.
\end{align*}
\end{small}We simulate the full dynamics by a modified next reaction method \cite{anderson2007modified} and the reduced model by Algorithm 2 in \cite{duncan2016hybrid}.

We assume that light intensity signals can be observed every 2 minutes from a fluorescent microscope 
and satisfy
\begin{equation*}
Y^{N,\gamma_1}(t)= h\left(X^{N,\gamma_1}(t)\right) + W(t)
\end{equation*}
where $h(x)= \left(10 \times x_4 \right) \wedge 10^{3}$ with $10^3$ being the measurement range, and $W(t)$ is a sequence of mutually independent standard Gaussian random variables.
Our goal is to provide estimates of latent dynamic states, including DNA copies, mRNA copies, and protein concentrations, based on readouts.



\begin{figure}[h!]
	\centering
	\includegraphics[width= 0.4 \textwidth]{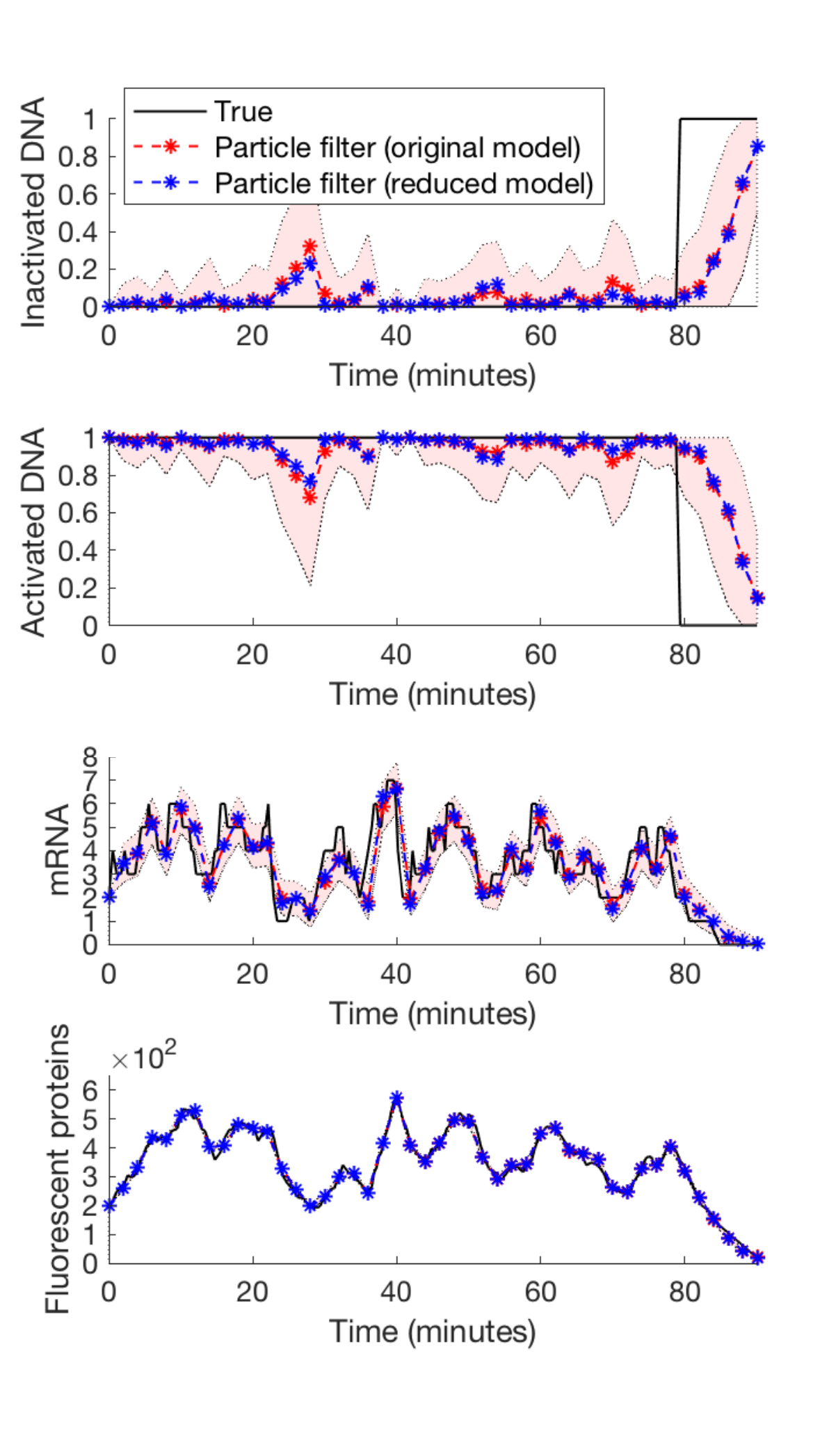}
	\caption{Comparison of different particle filters in the numerical experiment. The black line, red line, and blue line represent, respectively, the reference signal (the true signal), the particle filter $\bar\pi^{N,\gamma_1}_{M,t} (\phi)$ where the original model is applied in the sampling step, and the particle filter $\bar\pi^{\gamma_1}_{M,t} (\phi)$ where the reduced model is applied in the sampling step.
	Both particle filters adopt the multinomial branching method in the resampling step.	
    The light red region indicate the ``mean $\pm$ standard derivation" credible interval of the filter $\bar\pi^{N,\gamma_1}_{M,t} (\phi)$ (red line), which is computed by the particle filter using original model. 
		\label{fig numerical simulation at the first time scale}
	}
\end{figure}

A numerical simulation is shown in Fig \ref{fig numerical simulation at the first time scale}, where we first simulated the full dynamic system and observations, which were taken as true signals, and then applied two different types of particle filters to estimate the system. 
Here, we set the particle population to be 5000.
From the results, we can observe that both filters can follow the behavior of the true signal, and the particle filter applying the reduced model is within the credible interval (light red region) most of the time.
In the meanwhile, the particle filter applying the original model took 9.44 seconds in average to output the estimation in each iteration, whereas the particle filter  applying the reduced model took only 0.95 seconds in each iteration, 10 times faster than its opponent.
These results indicate that the proposed particle filter applying the reduced model is both accurate and computationally efficient in solving filtering problems for multi-scale stochastic chemical reaction network systems.

\section{Conclusion}{\label{Sec. conclusion}}

In this paper, we provide efficient particle filters to solve the filtering problem for multiscale stochastic reaction network systems, based on a time-scale separation technique.
We first show that the solution of the filtering problem for the original systems can be well approximated by the solution of the filtering problem for a reduced model that represents the dynamics as a hybrid process. 
Since the reduced model is based on exploiting the time-scale separation in the original network, the reduced model can greatly reduce the computational complexity to simulate the dynamics.
This enables us to develop efficient particle filters to solve the filtering problem for the original model by applying the particle filter to the reduced model.
Finally, a numerical example is presented to illustrate our approach.

There are a few topics deserving further investigation in future work.
First, since the variability of model parameters can greatly influence the performance of a filter, it is worth investigating if we can put parameter inference and latent state estimation into the same framework to improve the performance of the proposed particle filter.
Second, it is worth extending the convergence result of the particle filter (see Corollary \ref{cor particle filter}) to unbounded function scenarios, as many biological states are not necessarily bounded.
Third, we will also apply the proposed filter to the optimal control problem for an intracellular system.

\appendix

In this section, we give the outline of the proof of theorem \ref{thm convergence at the first time scale}.
Here, we utilize a framework proposed in \cite{calzolari2006approximation}, which require us to construct random variables $\tilde X^{N,\gamma_1}_{t_1:t_i}$, $\tilde X^{\gamma_1}_{t_1:t_i}$, $\tilde Y^{\gamma_1}_{t_1:t_i}$, $\tilde Z^{N,\gamma_1}(t_i)$ and $\tilde Z^{\gamma_1}(t_i)$ on a common probability space $\left(\tilde \Omega^{\gamma_1}, \tilde{\mathcal F}^{\gamma_1}, \tilde{ \mathbb Q}^{\gamma_1} \right)$ at each time point $t_i$ 
such that 
\begin{enumerate}[label=(A.\arabic*), itemindent=1em]
	\item  random variables $\left( \tilde  X^{N,\gamma_1}_{t_1:t_i}, \tilde  Y^{\gamma_1}_{t_1:t_i}, \tilde Z^{N,\gamma_1}(t_i)\right)$ take values in $\mathbb R^{n\times i} \times \mathbb R^{m\times i} \times \mathbb R$ and have the same law as $\left(  X^{N,\gamma_1}_{t_1:t_i},  Y^{N,\gamma_1}_{t_1:t_i}, Z^{N,\gamma_1}(t_i)\right)$ on $\left( \Omega, {\mathcal F}, {\mathbb{ P}}^{N,\gamma_1} \right)$. {\label{eq. assumption 1.1 of the big frame work}}
	\item the random variable $\left( \tilde  X^{\gamma_1}_{t_1:t_i}, \tilde  Y^{\gamma_1}_{t_1:t_i}, \tilde Z^{\gamma_1}(t_i)\right)$ takes values in $\mathbb R^{n\times i} \times \mathbb R^{m\times i} \times \mathbb R$ and has the same law as $\left(  X^{\gamma_1}_{t_1:t_i},  Y^{\gamma_1}_{t_1:t_i}, Z^{\gamma_1}(t_i)\right)$ on $\left( \Omega, {\mathcal F}, {\mathbb{ P}}^{\gamma_1} \right)$. {\label{eq. assumption 1.2 of the big frame work}}
	\item $\hat f_{\phi,t}^{N,\gamma_1}(\tilde Y^{\gamma_1}_{t_1:t_i}) \to  \hat f_{\phi,t}^{\gamma_1}(\tilde Y^{\gamma_1}_{t_1:t_i})$ in $\tilde{ \mathbb Q}^{\gamma_1}$-probability. \label{eq. assumption 1.3 of the big frame work}
	\item $\lim_{N\to \infty}\mathbb E_{\tilde{ \mathbb Q}^{\gamma_1}}\left[\left|\tilde Z^{N,\gamma_1}(t_i)-\tilde Z^{\gamma_1}(t_i)\right|\right]=0$. \label{eq. assumption 1.4 of the big frame work}
\end{enumerate}
If we succeed in finding the above random variables, then the convergence result is guaranteed by the following theorem.
\begin{theorem}[Adapted from \cite{calzolari2006approximation}] \label{theorem framework first time scale}
	For any $i> 0$ and any measurable function $\phi(\cdot)$, if \ref{eq. assumption 1.1 of the big frame work}, \ref{eq. assumption 1.2 of the big frame work}, \ref{eq. assumption 1.3 of the big frame work}, and \ref{eq. assumption 1.4 of the big frame work} are satisfied,
	then $\pi^{N,\gamma_1}_{t_i} (\phi) - \tilde \pi^{N,\gamma_1}_{t_i} (\phi)\to 0$ in $\mathbb{ P}$-probability.
\end{theorem}

We can construct such random variables and the common probability space as follows.
First, for any $i\in\mathbb N_{>0}$, we can apply the Skorokhod representation theorem to constructing  a common probability space $\left(\tilde \Omega_1^{\gamma_1}, \tilde{\mathcal F}_1^{\gamma_1}, \tilde{\mathbb Q}_1^{\gamma_1} \right)$ on which are defined $\mathbb R^{n\times i}$-valued random variables $\tilde X^{N,\gamma_1}_{t_1:t_i}$ and $\tilde X^{\gamma_1}_{t_1:t_i}$ with the laws of $ X^{N,\gamma_1}_{t_1:t_i}$ and $ X^{\gamma_1}_{t_1:t_i}$, respectively, such that $\lim_{N\to\infty} \tilde X^{N,\gamma_1}_{t_1:t_i}= \tilde X^{\gamma_1}_{t_1:t_i}$ $\tilde{\mathbb Q}_1^{\gamma_1}$-almost surely.
Then, we term $\left(\tilde \Omega_2^{\gamma_1}, \tilde{\mathcal F}_2^{\gamma_1}, \tilde{\mathbb Q}_2^{\gamma_1} \right)$ as a copy of the probability space $\left( \Omega, {\mathcal F}_{t_i}, \mathbb P^{\gamma_1}\right)$ and $\tilde Y^{\gamma_1}_{t_1:t_i} $ as a copy of $Y^{\gamma_1}_{t_1:t_i}$ in this probability space.
Moreover, we define $\tilde Z^{N\gamma_1}(t_i)\triangleq \prod_{j=1}^{i} g\left(\tilde X^{N\gamma_1}(t_j),\tilde Y^{\gamma_1}(t_j)\right)$ and $\tilde Z^{\gamma_1}(t_i)\triangleq \prod_{j=1}^{i} g\left(\tilde X^{\gamma_1}(t_j),\tilde Y^{\gamma_1}(t_j)\right)$ on the product probability space $\left(\tilde \Omega^{\gamma_1}, \tilde{\mathcal F}^{\gamma_1}, \tilde{\mathbb Q}^{\gamma_1} \right) \triangleq \left(\tilde \Omega_1^{\gamma_1}\times\tilde \Omega_2^{\gamma_1}, \tilde{\mathcal F}_1^{\gamma_1} \times \tilde{\mathcal F}_2^{\gamma_1},\tilde{\mathbb Q}_1^{\gamma_1}\times \tilde{\mathbb Q}_2^{\gamma_1} \right)$, 
where  $\tilde X^{N,\gamma_1}(t_j)$ is the $j$-th column of $\tilde X^{N,\gamma_1}_{t_1:t_j}$, $\tilde X^{\gamma_1}(t_j)$ is the $j$-th column of $\tilde X^{\gamma_1}_{t_1:t_j}$, and $\tilde Y^{\gamma_1}(t_j)$ is the $j$-th column of $\tilde Y^{\gamma_1}_{t_1:t_j}$.

Obviously, such defined random variables and the common probability space satisfy \ref{eq. assumption 1.1 of the big frame work} and \ref{eq. assumption 1.2 of the big frame work}.
The condition \ref{eq. assumption 1.4 of the big frame work} can be proven straightforwardly by the fact that $\lim_{N\to\infty} \tilde X^{N,\gamma_1}_{t_1:t_i}= \tilde X^{\gamma_1}_{t_1:t_i}$ $\tilde{\mathbb Q}_1^{\gamma_1}$-almost surely.
Similarly, we can also show \ref{eq. assumption 1.3 of the big frame work} by using \eqref{eq. unnormalized filter to normalized filter} and \eqref{eq. unnormalized filters reduced model}. 
Consequently, Theorem \ref{thm convergence at the first time scale} is proven.

\end{document}